\documentclass[conference,letterpaper]{IEEEtran} 

\usepackage{amsthm,amssymb,graphicx,multirow,amsmath,color,amsfonts,balance}
\usepackage{bm}
\usepackage{physics}

\usepackage[latin1]{inputenc}
\usepackage[update,prepend]{epstopdf}
\usepackage[noadjust]{cite}
\usepackage{mathtools}
\usepackage{multirow}
\usepackage{bbm} 
\usepackage{pdfpages}
\usepackage{tabulary}
\usepackage{multirow}
\usepackage{nicefrac}
\usepackage{comment}
\usepackage{algorithm}
\usepackage{algorithmicx}
\usepackage{algpseudocode}
\usepackage{etoolbox}
\usepackage{enumerate}
\usepackage{subcaption}
\usepackage[svgnames]{xcolor}
\def\beq{\begin{equation}}
\def\eeq{\end{equation}}
\def\beqa{\begin{eqnarray}}
\def\eeqa{\end{eqnarray}}
\def\beqan{\begin{eqnarray*}}
\def\eeqan{\end{eqnarray*}}
 \usepackage{float}
\usepackage{stfloats}

\usepackage{tikz}
\usetikzlibrary{chains,arrows,calc,positioning}
\usetikzlibrary{shapes.multipart}
\usetikzlibrary{decorations.pathreplacing}
\usetikzlibrary{calc}

\usepackage[normalem]{ulem} 

\newcommand{\citep}[1]{\cite{#1}}

\renewcommand{\Re}{\operatorname{Re}}

\renewcommand{\tilde}{\widetilde}
\renewcommand{\hat}{\widehat}

\def\beq{\begin{equation}}
\def\eeq{\end{equation}}
\def\beqa{\begin{eqnarray}}
\def\eeqa{\end{eqnarray}}
\def\beqan{\begin{eqnarray*}}
\def\eeqan{\end{eqnarray*}}

\def\C{{\mathbb{C}}}

\DeclareMathOperator{\diag}{Diag}

\newtheorem{theorem}{Theorem}
\newtheorem{proposition}{Proposition}

\newtheorem{lemma}{Lemma}

\theoremstyle{definition}

\def\Exp{\mathbb{E}}

\def\Tr{\mathrm{Tr}}
\def\tr{\mathrm{Tr}}

\newcommand\blfootnote[1]{%
  \begingroup
  \renewcommand\thefootnote{}\footnote{#1}%
  \addtocounter{footnote}{-1}%
  \endgroup
}

\newcommand{\zero}{\mathbf{0}}

\newcommand{\abf}{\mathbf{a}}
\newcommand{\bbf}{\mathbf{b}}
\newcommand{\cbf}{\mathbf{c}}

\newcommand{\ebf}{\mathbf{e}}

\newcommand{\gbf}{\mathbf{g}}
\newcommand{\hbf}{\mathbf{h}}
\newcommand{\pbf}{\mathbf{p}}

\newcommand{\rbf}{\mathbf{r}}

\newcommand{\ubf}{\mathbf{u}}
\newcommand{\ubfhat}{\widehat{\mathbf{u}}}

\newcommand{\xbf}{\mathbf{x}}

\newcommand{\ybf}{\mathbf{y}}

\newcommand{\zbf}{\mathbf{z}}

\newcommand{\zbfhat}{\widehat{\mathbf{z}}}
\newcommand{\Abf}{\mathbf{A}}
\newcommand{\Bbf}{\mathbf{B}}
\newcommand{\mubf}{\mathbf{\Mu}}
\newcommand{\lbf}{\boldsymbol \ell}

\newcommand{\Cbf}{\mathbf{C}}

\newcommand{\Fbf}{\mathbf{F}}
\newcommand{\Gbf}{\mathbf{G}}
\newcommand{\Hbf}{\mathbf{H}}

\newcommand{\Ibf}{\mathbf{I}}
\newcommand{\Jbf}{\mathbf{J}}

\newcommand{\Pbf}{\mathbf{P}}

\newcommand{\Qbf}{\mathbf{Q}}
\newcommand{\Rbf}{\mathbf{R}}

\newcommand{\Sbf}{\mathbf{S}}
\newcommand{\Tbf}{\mathbf{T}}
\newcommand{\tbf}{\mathbf{t}}
\newcommand{\Ubf}{\mathbf{U}}

\newcommand{\Vbf}{\mathbf{V}}

\newcommand{\Wbf}{\mathbf{W}}

\newcommand{\Xbf}{\mathbf{X}}
\newcommand{\Ybf}{\mathbf{Y}}
\newcommand{\Zbf}{\mathbf{Z}}

\newcommand{\Zbfhat}{\widehat{\mathbf{Z}}}

\def\nubf{{\boldsymbol \nu}}
\def\mubf{{\boldsymbol \mu}}

\def\nubf{{\boldsymbol \nu}}
\def\xibf{{\boldsymbol \xi}}
\def\Xibf{{\boldsymbol \Xi}}

\newcommand{\Phibf}{{\boldsymbol \Phi}}

\newcommand{\indic}[1]{\mathbbm{1}_{ \{ {#1} \} }}

\newcommand{\tran}{^{\text{\sf T}}}

\newcommand{\str}{^{\text{\sf *}}}
\newcommand{\herm}{^{\text{\sf H}}}

\newcommand{\bkt}[1]{{\left< #1 \right>}}

\begin{document}

\bstctlcite{IEEEexample:BSTcontrol}
\title{Quantized MIMO: Channel Capacity and Spectrospatial Power Distribution}

\author{
 \IEEEauthorblockN{
 Abbas Khalili, Elza Erkip, Sundeep Rangan} 

\IEEEauthorblockA{Dept. of Electrical and Computer Engineering, NYU Tandon School of Engineering.}
}
\maketitle

\begin{abstract}
Millimeter wave systems suffer from high power consumption and are constrained to use low resolution quantizers \textemdash  digital to analog and analog to digital converters (DACs and ADCs). However, low resolution quantization leads to reduced data rate and increased out-of-band emission noise. In this paper, a multiple-input multiple-output (MIMO) system with linear transceivers using low resolution DACs and ADCs is considered. An information-theoretic analysis of the system to model the effect of quantization on spectrospatial power distribution and capacity of the system is provided. More precisely, it is shown that the impact of quantization can be accurately described via a linear model with additive independent Gaussian noise. This model in turn leads to simple and intuitive expressions for spectrospatial power distribution of the transmitter and a lower bound on the achievable rate of the system. Furthermore, the derived model is validated through simulations and numerical evaluations, where it is shown to accurately predict both spectral and spatial power distributions. \blfootnote{ The work  supported in part by
NSF grants 1952180, 1925079, 1564142, 1547332, SRC, and the industrial affiliates of NYU WIRELESS.}
\end{abstract}

\section{Introduction}

Digital to analog and analog to digital converters (DACs and ADCs) are essential components of any digital communication system. While sub-$6$GHz systems use high resolution DACs and ADCs, millimeter wave (mmWave) rely on communication across wide bandwidths with large numbers of antennas thereby making high resolution DACs and ADCs very costly in terms of the power consumption \cite{rappaport2014millimeter,rangan2014millimeter}. Use of low resolution DACs and ADCs (typically 3-4 bits in I and Q) have been suggested as energy-efficient approaches for next-generation mmWave systems
\cite{abbas2017millimeter,zhang2018low,abdelghany2018towards,abbasISIT2018,abbas2019highsnr,abbas2019MT,Abbas2020Thr,singh2009limits,koch2013low,nossek2006capacity,orhan2015low,dutta2019case, mo2015capacity,rini2017general,mezghani2012capacity,Studer2016,jacobsson2017throughput,mollen2016uplink,mo2017hybrid}

Unlike its high resolution counterpart, low resolution quantization leads to high quantization noise reducing the achievable rate of the system
and adding out-of-band (OOB) emission noise which results in high adjacent carrier leakage ratio (ACLR). 
Therefore, it is of great importance to model this quantization noise and understand its impacts accurately. 

There is a large body of work on low resolution DACs and ADCs \cite{abbas2017millimeter,zhang2018low,abdelghany2018towards,khalili2021mimo,abbasISIT2018,abbas2019highsnr,abbas2019MT,Abbas2020Thr,singh2009limits,koch2013low,nossek2006capacity,orhan2015low,dutta2019case, mo2015capacity,rini2017general,mezghani2012capacity,Studer2016,jacobsson2017throughput,mollen2016uplink,mo2017hybrid,skrimponis2022understanding,skrimponis2020towards}.
The most common model used to approximate the effect of quantization is the additive Gaussian noise (AGN) model  \cite{gersho2012vector,fletcher2007robust}. This model has been rigorously analyzed in several works in the high rate quantization regime or under dithered quantization
\cite{gersho2012vector,marco2005validity,gray1993dithered,zamir1996LQN,derpich2008quadratic}. For low resolution mmWave systems, it is observed that the AGN and other 
Gaussian noise predictions match simulations and therefore provide an accurate model of quantization; however, the accuracy is not rigorously justified 
\cite{mo2015capacity,rini2017general,mezghani2012capacity,Studer2016,jacobsson2017throughput,mollen2016uplink,mo2017hybrid}.

In \cite{dutta2020capacity}, we considered a single-input single-output (SISO) communication system with a linear transceiver and low resolution quantization, and using information theoretic tools, proved the validity of the AGN model. In this paper,  we extend the analysis of \cite{dutta2020capacity} to multiple-input multiple-output (MIMO) communications. More precisely, we consider a MIMO system using hybrid transceivers with low resolution DACs and ADCs as illustrated in Fig.~\ref{fig:Sys_Model}. The transmitter modulates a set of transmit streams through a unitary transform $\Vbf \herm$ prior to the DACs. We model the spectrum of the transmitted signal through the transform $\rbf = \Vbf\xbf$, where $\xbf$ is the output of the transmitter. Note that if we set the unitary matrix $\Vbf$ to an FFT-matrix, we would have a MIMO-OFDM system and can capture filtering and spectral constraints which arise form practical constraint. The MIMO extension of \cite{dutta2020capacity} presented here requires new results on empirical convergence and mutual information of two random matrices. In addition, it allows us to study spatial distribution of power as argued below. A summary of our contributions are as follows:

\begin{itemize}
    \item \emph{Rigorous linear model:} 
    We analyse a certain large random limit of the system where $\Vbf \in \C^{N \times N}$ is selected uniformly from $N\times N$ unitary matrices and provide a rigorous framework for a linear model of quantization.
    This generalizes our result in \cite{dutta2020capacity} which considers only SISO systems (Sec.~\ref{sec:emplin}).

    \item \emph{Rate and spectrospatial power distribution:}
    We use the derived linear model and provide a simple asymptotically
    exact expression for per sample spectral covariance matrix of the system in Fig. \ref{fig:Sys_Model}. Such covariance matrices can be computed at both transmitter or receiver leading to spectrospatial power distributions and lower bounds on the system's capacity (Sec.~\ref{sec:lsl}). Furthermore, through simulations and numerical evaluations we show that the derived linear model accurately predicts ACLR and spatial power distribution (Sec.~\ref{sec:num_res}).
\end{itemize}
\emph{Notation}: $x$ is scalar, $\xbf$ is a column vector whose $i^{\rm th}$ element is $x_i$, and $\Xbf$ is a matrix whose $(i,j)^{\rm th}$ element is $x_{ij}$. We denote the column $i$ of $\Xbf$ with $\xbf_{:i}$ and with a slight abuse of notation, we denote the transpose of row $j$ with $\xbf_{j:}$. Also, $X$ is a random column vector and $[N]$ is the set $\{1,2,\ldots N\}$.

\section{System Model and Preliminaries}
\begin{figure*}[tp]
    \centering
    \includegraphics[width = \textwidth]{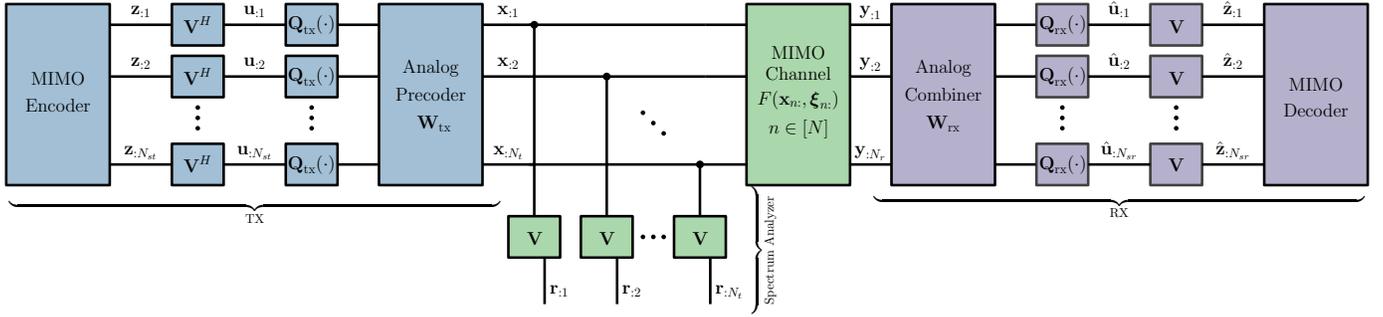}
     \caption{System model with linear modulation, demodulation and quantization at both the transmitter and receiver. The linear modulation is modeled as a multiplication by $\Vbf\herm$ prior to quantization at the transmitter, while a spectrum analyzer and receiver employ the inverse transform $\Vbf$.}
    \label{fig:Sys_Model}
    \vspace{-0.3cm}
\end{figure*}
\subsection{Transceiver with Linear Transform}
We consider the general transceiver structure with linear transform modulation and analog precoder and combiner shown in Fig. \ref{fig:Sys_Model}. In this system, the transmitter generates $N_{st}$ streams of length $N$ denoted by matrix $\Zbf$, whose columns $\zbf_{:i} = [z_{1i}, z_{2i}, \cdots, z_{Ni}]\tran \in \mathbb{C}^{N}, i \in [N_{st}]$ are linearly modulated as $\ubf_{:i} = \Vbf \herm \zbf_{:i}$, where $\Vbf \in \mathbb{C}^{N \times N}$ is a unitary matrix. The modulated signals are converted into analog using a set of DACs $\Qbf_{\rm tx}(\cdot)$, and then passed through a linear analog precoder determined by the matrix $\Wbf_{{\rm tx}} \in \mathbb{C}^{N_t \times N_{st}}$ and transmitted. If $V$ were an FFT matrix, we could consider the symbols $z_{ni}, i \in [N_{st}], n\in[N]$ as the values of the information carrying signal in frequency domain and $u_{ni}, i \in [N_{st}], n\in[N]$ the digital values in time-domain. The modulation can thus be regarded as a simplified version of the OFDM (where we  ignore  the  cyclic  prefix). In  addition, if we zero-pad the input frequency-domain symbols $z_{ni}, i \in [N_{st}], n\in [N]$, the transformed vector $\ubf_{:i} = \Vbf \herm \zbf_{:i}$ can be seen as a linearly up-sampled version of $\zbf_{:i}$.
The transmitted signal goes through a general MIMO channel of the form
\begin{align}
\label{eq:ch}
    \Ybf = \Fbf(\Xbf, \Xibf),
\end{align}
where $\Xbf = [\xbf_{:1},\xbf_{:2}, \cdots, \xbf_{:N_t}]$, $\Ybf = [\ybf_{:1},\ybf_{:2}, \cdots, \ybf_{:N_r}]$, $\Xibf \in \mathbb{C}^{N \times N_r}$ with $\xi_{ni} \sim \mathcal{CN}(0, 1)$, and $\Fbf(\cdot)$ is a row-wise function. This channel can also model certain non-linearites in the RF front-end \cite{abdelghany2018towards}. For the case of MIMO AWGN channel, 
\begin{align}
\label{eq:awgn}
    \ybf_{n:} = \Hbf\xbf_{n:} + \xibf_{n:}, n\in[N],
\end{align}
where $\ybf_{n:}\tran, \xbf_{n:} \tran$, and $\xibf_{n:}\tran$ are the  $n^{\rm th }$ rows of matrices $\Ybf$, $\Xbf$, and $\Xibf$, respectively.
At the receiver, the channel output is first passed through a linear analog combiner  $\Wbf_{{\rm rx}}\in \mathbb{C}^{N_{sr} \times N_{r}}$, and quantized using a set of ADCs $\Qbf_{\rm rx}(\cdot)$, and then the receiver performs the inverse transform operation to obtain $\zbfhat_{:j} = \Vbf \ubfhat_{:j}, j \in [N_{sr}]$. Let $\Zbfhat = [\zbfhat_{:1},\zbfhat_{:2}, \cdots, \zbfhat_{:N_{sr}}]$ using \eqref{eq:ch}, we have
\begin{subequations}
\label{eq:sys_mod}
\begin{gather}
    \Zbfhat = \Vbf \Gbf(\Vbf \herm \Zbf, \Xibf),\\
    \Gbf(\Ubf, \Xibf) = \Qbf_{\rm rx}\left( \Fbf\left(\mathbf\Qbf_{\rm tx} \left(\Ubf\right) \Wbf_{{\rm tx}}\tran, \Xibf\right) \right)\Wbf_{{\rm rx}}\tran.
\end{gather}
\end{subequations}

\subsection{Spectral Covariance and Spectrospatial Power}
\label{subsec:ssp}
To derive per sample spectral covariance matrix at the transmitter, consider $\rbf_{:k} = \Vbf\xbf_{:k}, k\in[N_{t}]$ which is the transform of the transmitted signal from the $k^{\rm th}$ antenna.
The component $|r_{nk}|^2, n\in[N]$ can be regarded as the energy of the transmit signal from $k^{\rm th}$ antenna at frequency $n$.

We assume the frequency is divided into $M$ sub-bands.
Let $\ell_{n} \in [M]$ be the variable that indicates
which sub-band frequency $n$ of the transmitted signals (rows of matrix $\Xbf$) belong to.
We call $\lbf=(\ell_1,\ldots,\ell_{N})$ the \emph{sub-band
selection vector} and define,
\beq \label{eq:delm}
    \delta_m(\lbf) := \frac{1}{N} \sum_{n=1}^{N} \indic{\ell_n = m},
\eeq
which represents the fraction of the frequency components in
sub-band $m$.  
We will call $\delta_m$ the \emph{bandwidth fraction}
for sub-band $m$. Define $\Rbf = [\rbf_{:1},\rbf_{:2}, \cdots, \rbf_{:N_t}] = [\rbf_{1:},\rbf_{2:}, \cdots, \rbf_{N:}]\tran$. The per \!sample covariance matrix in sub-band $m$~is
\begin{align} 
\label{eq:sbfm}
    \Sbf_m(N):= \frac{1}{N}\sum_{n=1}^{N}  \rbf_{n:} \rbf_{n:}\herm \indic{l_n=m},
\end{align}

As a result, for sub-band $m$, the fraction of spectral power, $\nu_m$, and spatial (beamforming) power towards angle $\psi \in (0,2\pi]$, ${\rm BF}_m(\psi)$, respectively, are
\begin{gather}
\label{eq:spr}
    \nu_m = \frac{\tr\left(\Sbf_m(N)\right)}{\tr\left(\Sbf(N)\right)}, ~\quad \Sbf(N) = \sum_m \Sbf_m(N), ~\mathrm{and}\\
\label{eq:bfg}
     {\rm BF}_m(\psi)= \ebf(\psi)\herm \Sbf_m(N) \ebf(\psi), 
\end{gather}
where $\ebf(\psi)$ is the transmit array response. For the case of uniform linear array (ULA) with half wavelength spacing $\ebf(\psi) = [1, e^{i\pi \cos(\psi)}, \ldots, e^{(N_t-1)i\pi \cos(\psi)}]\tran$. Note that the analysis can also be applied to two dimensional beamforming.

\subsection{Achievable Rate}

An achievable rate for the system can be computed by fixing the distributions of $\zbf_{:i}, i\in [N_{st}]$ and calculating the mutual information $I(\Zbf; \Zbfhat)$ between the transmitted and received frequency-domain matrices $\Zbf$, $\Zbfhat$.
For the input distribution, we will use an independent complex Gaussian in
each frequency and stream.  Specifically, we will assume the components
$\zbf_{n:}, n\in [N]$ ( $\zbf_{n:}\tran$ is the $n^{\rm th}$ row of matrix $\Zbf$) are i.i.d.
\beq \label{eq:zmix}
    \zbf_{n:} \sim C{\mathcal N}(\zero,\Pbf_{m}), \mbox{ when }
    \ell_n = m,
\eeq
where $\Pbf_{m}$ is the covariance matrix of the components in sub-band $m$. As a result, the average per symbol (row) covariance matrix for $\Zbf$ is
\beq \label{eq:Pbar}
    \Pbf = \frac{1}{N} \Exp \left[\Zbf\tran \Zbf\str \right] = 
    \frac{1}{N} \sum_{n =1}^N \Exp\left[\zbf_{n:} \zbf_{n:}\herm \right] = \sum_{m=1}^M \delta_m\Pbf_{m},
\eeq
where $\delta_m$s are the bandwidth fractions \eqref{eq:delm}.
We note that using the Gaussian input distribution is not necessarily optimal in systems with quantization. Finding the optimal input distribution is left for future work.

\subsection{Assumptions}
\label{subsec:assump}
For tractability of the analysis, we consider a certain large system limit
of random instances of the system indexed by the dimension $N$ with $N \rightarrow \infty$. 
For each $N$, We consider $\Vbf=\Vbf(N)$ is a random unitary matrix $\Vbf$ that is uniformly distributed on the $N \times N$ unitary matrices (i.e.\ Haar distributed) instead of considering the deterministic
FFT matrix. We assume that the sub-band selection vectors  $\lbf=\lbf(N)$ are fixed sequence satisfying,
\beq \label{eq:delm_limit}
    \lim_{N \rightarrow \infty} \frac{1}{N} |\{ \ell_n(N) = m \}| = \delta_m.
\eeq
This condition ensures that
a fraction $\delta_m$ of the components are in sub-band $m$.

We consider that the channel $\Fbf(\cdot)$ is  Lipschitz continuous and acts as a \emph{row-wise separable} function. More specifically, 
\beq \label{eq:Chcomp}
    \Ybf = \Fbf(\Xbf, \Xibf) \Longleftrightarrow \ybf_{n:} = F(\xbf_{n:} , \xibf_{n:}),
\eeq
where $F(\cdot)$ is a Lipschitz vector-input, vector-output function.
Note that these conditions are satisfied for the MIMO AWGN channel in \eqref{eq:awgn} as it is performs the same operation on each row of the pair $(\Xbf,\xibf)$ and has the Lipschitz constant equal to the maximum singular value of $\Hbf$. 
Similarly, we assume that the DAC and ADC functions, $\Qbf_{\rm x}(\ubf) , {\rm x}\in\{{\rm tx, rx}\}$ are Lipschitz continuous
and \emph{component-wise separable} for some scalar-input, scalar-output function $Q_{\rm x}(\cdot),  {\rm x}\in\{{\rm tx, rx}\}$, respectively. We also  assume that these functions are deterministic and do not change over time.  We note that typical quantizers are not Lipschitz continuous. However, we assume they can be approximated arbitrarily closely by a Lipschitz function. Through simulations, we will show that the predictions hold true even
for standard quantizers.

For our proofs, we use results on empirical convergence. The analysis framework was developed by Bayati and Montanari \cite{bayati2011dynamics} and also used in the vector approximate massage passing (VAMP) analysis \cite{rangan2019vector}. In next section, we will provide a definition of empirical convergence along with the necessary result for our proofs.

\section{Empirical Convergence of Random Vectors}
\label{sec:emplin}
For a given $p\geq 1$, a mapping $\Phi:\C^{d_i}\rightarrow \C^{d_o}$ is called \emph{pseudo-Lipschitz} of order $p$ if
\begin{align*}
\|\Phi(\xbf_1)-\Phi(\xbf_2)\|\!\leq \!C\|\xbf_1-\xbf_2\|\left(1\!+\!\|\xbf_1\|^{p-1}\!+\!\|\xbf_2\|^{p-1}\right),\!
\end{align*} 
for some constant $C > 0$.  Note that when $p=1$, we obtain the standard definition of Lipschitz continuity. Now, suppose that for each $N$, $\Xbf(N)$ is a matrix 
$\Xbf(N) = [\xbf_{1:},\ldots,\xbf_{N:}]^{\tran}$ with rows $\xbf_{n:}\tran \in \C^{1\times d_i}$ for some fixed dimension $d_i$.
Let $X \in \C^{d_i}$ be a random vector. We say that the rows of $\Xbf(N)$ \emph{converge empirically to $X\tran$ with $p$-th
order moments} if
\beq \label{eq:plconv}
    \lim_{N \rightarrow \infty} \frac{1}{N} \sum_{n=1}^{N} \Phi(\xbf_{n:}(N)) = \Exp\left[ \Phi(X) \right],
\eeq
for all pseduo-Lipschitz functions of order $p$.  Loosely speaking, the condition requires that the empirical distribution of the rows of $\Xbf(N)$ converge to that of the random vector $X\tran$ which is satisfied when $\xbf_{n:}$s are i.i.d. with distribution $X$.
We will denote this by
\beq \label{eq:plp}
    \lim_{N \rightarrow \infty} \{ \xbf_{n:}\} \stackrel{PL(p)}{=} X.
\eeq
Next proposition describes the distribution of the matrices under random unitary transform and generalizes \cite[Prop.~1]{dutta2020capacity} which only considers the vector case. Let us consider a sequence of systems indexed by $N$, and for each $N$ suppose that 
$\Vbf \in \C^{N \times N}$ is uniformly distributed on the unitary matrices.
Let $\Xbf =\Xbf(N)$ be a sequence of matrices with 
\beq \label{eq:xs_gen}
    \lim_{N \rightarrow \infty} \{ \xbf_{n:} \} \stackrel{PL(2)}{=} X.
\eeq
\begin{proposition}
\label{prop:matdist}
Define $\Ubf = \Vbf \Xbf$. Given \eqref{eq:xs_gen} and Haar distributed matrix $\Vbf$, the sequence $\{\ubf_{n:}\}$ converge empirically to random vector $U \sim C\mathcal{N}(\zero, \Pbf)$, with $\Pbf = \Exp\left[ X  X \herm \right]$.
\end{proposition}
\begin{proof}
See Appendix \ref{app:matdist}.   
\end{proof}
This proposition shows that random unitary transformation effectively creates a Gaussian distribution in the sense that for input matrices whose rows converge empirically, the rows of the output matrix converge empirically to a Gaussian distribution. Now, consider a matrix $\Ybf$ generated by,
\beq \label{eq:yG_gen}
    \Ybf = \Vbf \Phibf( \Vbf\herm \Xbf,\Xibf),
\eeq
where $\Phibf(\cdot)$ is some function that operates row-wise in that
\beq
\label{eq:phicon}
    \Ybf = \Phibf(\Ubf, \Xibf) \Longleftrightarrow \ybf_{n:} = \Phi(\ubf_{n:}, \xibf_{n:}),
\eeq
where $\Phi(\cdot)$ is a pseudo-Lipschitz vector-input, vector-output function.
Assume that rows of $\Xibf$ also converge empirically in that $\lim_{N \rightarrow \infty} \{ \xibf_{n:} \} \stackrel{PL(2)}{=} \Xi,$
for some random vector $\Xi$.
To analyze the statistics on $\Ybf$, we define the quantities:
\begin{subequations} \label{eq:tau_alpha_gen}
\begin{gather}
    \Abf := \Exp\left[\frac{\partial \Phi(U,\Xi)}{\partial U} \right], \\
    \Tbf := \Exp\left[(\Phi(U,\Xi)- \Abf U ) (\Phi(U,\Xi)- \Abf U  ) \herm\right].
\end{gather}
\end{subequations}
where $U \sim \mathcal{CN}(\zero, \Pbf)$, $\Pbf := \Exp\left[X X \herm\right]$, and $\nicefrac{\partial \Phi(U,\Xi)}{\partial U}$ denotes the Jacobian of $\Phi(\cdot)$ with respect to the vector $U$. To calculate the matrix $\Abf$ one can use the multivariate extension of Stein's lemma provided in \cite{landsman2008stein} i.e.,
\beq
\Pbf\Abf \herm = \Exp \left[  U\Phi(U,\Xi)\herm \right].
\eeq
\begin{proposition}  
\label{prop:lin}
Given the relation in \eqref{eq:yG_gen} and the conditions \eqref{eq:xs_gen} and \eqref{eq:phicon}, the rows of the matrix pair $(\Ybf,\Xbf)$ converge empirically as,
\beq \label{eq:lin_model_lim}
    \lim_{N \rightarrow \infty} \{(\ybf_{n:} ,\xbf_{n:}) \} \stackrel{PL(2)}{=} (Y,X),
\eeq
and 
\beq \label{eq:lin_model_gen}
    Y = \Abf X + \Theta, \quad \Theta \sim C{\mathcal N}(\zero,\Tbf),
\eeq
with $\Theta$ independent of $X$.
\end{proposition}
\begin{proof}
See Appendix \ref{pr:emp_conv}.
\end{proof}

The model \eqref{eq:lin_model_gen} shows that transformation of $\Xbf$ to produce $\Ybf$ results in a linearly scaled $\Xbf$ plus Gaussian noise. The scaling matrix $\Abf$ and Gaussian noise covariance matrix $\Tbf$ can be computed from the distributions of the components as shown in \eqref{eq:tau_alpha_gen}. By substituting the quantization function $Q_{\rm tx}(\cdot)$ in the place of the function $\Phi(\cdot)$,
we observe that the time-domain quantization effectively scales the frequency signal $\zbf$ and adds an independent Gaussian noise. 

\section{Achievable Spectral Energy and Rate} 
\label{sec:lsl}

\subsection{Spectral Covariance}
We first compute the asymptotic per sample spectral covariance matrix shown in \eqref{eq:sbfm}. Given the assumptions in Sec.~\ref{subsec:assump}, calculate $\Abf_{\rm tx}$ and $\Tbf_{\rm tx}$ using \eqref{eq:tau_alpha_gen} with $\Phi(U) = \Wbf_{\rm tx}Q_{\rm tx}(U)$, where $U \sim C\mathcal{N}(\zero, \Pbf),  ~ \Pbf = \sum_{m}\delta_m \Pbf_m$.

\begin{theorem}  
\label{thm:spec_fd}
Let $\Rbf=\Vbf\Xbf$ be the frequency-domain representation of the transmitted signal $\Xbf$.
Then the covariance matrix in each sub-band converges almost surely to
\begin{align} 
    \Sbf_m := \lim_{N \rightarrow \infty} \Sbf_m(N)
    =\delta_m\left[  \Abf_{\rm tx}\Pbf_m\Abf_{\rm tx} \herm + \Tbf_{\rm tx} \right].
        \label{eq:sm_fd}
\end{align}
In particular, the total covariance matrix per sample converges almost
surely as,
\begin{align} 
    \Sbf &:= 
    \lim_{N \rightarrow \infty} \frac{1}{N} \Xbf\tran \Xbf \str= \Abf_{\rm tx}\Pbf\Abf_{\rm tx} \herm + \Tbf_{\rm tx}
        \label{eq:stot_fd}
\end{align}
\end{theorem}
\begin{proof}
See Appendix \ref{sec:spec_proof}. 
\end{proof}
Based on Thm.~\ref{thm:spec_fd}, spectrospatial power distribution of the transmitter can be computed using \eqref{eq:spr} and \eqref{eq:bfg}. Based on \eqref{eq:spr}, the fraction of power in sub-band $m$ is
\beq \label{eq:nu_fd}
    \nu_m := \frac{\tr(\Sbf_m)}{\tr(\Sbf)} = 
        \frac{\delta_m\tr\left(  \Abf_{\rm tx}\Pbf_m\Abf_{\rm tx} \herm + \Tbf_{\rm tx} \right)}{\tr \left( \Abf_{\rm tx}\Pbf\Abf_{\rm tx} \herm + \Tbf_{\rm tx} \right)}.
\eeq
As a result, we always have
\beq \label{eq:nu_fd_feas}
    \nu_m \geq 
        \frac{\delta_m\tr\left(\Tbf_{\rm tx} \right)}{\tr \left( \Abf_{\rm tx}\Pbf\Abf_{\rm tx} \herm + \Tbf_{\rm tx} \right)}.
\eeq
In fact, the converse is also true. More precisely, in the next proposition we show that for a given pair $(\Abf_{\rm tx}, \Tbf_{\rm tx})$ and input covariance matrix $\Pbf$, there exist per sub-band covariance matrices $\Pbf_m$ resulting in a power fraction vector $\nubf = (\nu_1,\ldots,\nu_M)$ 
if and only if $\nu_m \geq 0$, $\sum_m \nu_m = 1$, and \eqref{eq:nu_fd_feas} is satisfied. 

\begin{proposition} \label{prop:nu_feas}
Let $\Abf_{\rm tx} \in \mathbb{C}^{N_t\times N_{st}}$, $\tr(\Tbf_{\rm tx}) >  0$,
$\Pbf$ with $\diag(\Pbf)>0$ and
$\delta_m \geq 0$ with $\sum_m \delta_m = 1$ be given.
For any $\nubf=(\nu_1,\ldots,\nu_M)$, 
the following are equivalent:
\begin{enumerate}[(a)]
\item For a set of $\nu_m$s given by \eqref{eq:nu_fd}, there exists $\Pbf_m $ with $\diag(\Pbf_m)\geq 0$ such that 
$\Pbf = \sum_m \delta_m \Pbf_m$.
\item $\nu_m$ satisfies \eqref{eq:nu_fd_feas} for all $m$
and $\sum_m \nu_m = 1$.
\end{enumerate}
\end{proposition}
\begin{proof}
See Appendix~\ref{sec:linear_feas}.
\end{proof}
Note that \eqref{eq:nu_fd_feas} shows 
that the power in a sub-band cannot be reduced less than a threshold. 
This arises from the fact that the
quantization noise of a quantizer is white and places power across the
spectrum. Similar observation was made for the case of SISO systems in \cite{dutta2020capacity}. This suggests that, depending on the system parameters such as DAC resolution, there is a limit on how much OOB emission can be reduced. This is of great importance in practical systems with regulations on the OOB emission levels.

\subsection{Achievable Rate}
We next lower bound the achievable rate of the system
\begin{align} \label{eq:lin_rate_def}
    R_{\rm lin} := \liminf_{N \rightarrow \infty} \frac{1}{N} I(\Zbf;\Zbfhat),
\end{align}
between the transmitted symbols $\Zbf$ and received symbols $\Zbfhat$. We refer to this as the \emph{linear rate} since it is achieved using linear modulation as in Fig.~\ref{fig:Sys_Model}.
To bound \eqref{eq:lin_rate_def}, we use the following result which is a non-trivial extension of \cite[Lemma~1]{dutta2020capacity} derived for the vector case. 
\begin{lemma} Suppose that $\Zbf \in \C^{N_{st}\times d}$ is a random matrix 
with i.i.d. columns $Z \sim  C{\mathcal N}(\mubf_{\zbf},
\Pbf)$.  Let $\Ybf$ be another random matrix and define,
\begin{subequations}
\vspace{-0.3cm}
\begin{gather}
    \Pbf_C := \Pbf^{-1}\frac{1}{d}\sum_{ i=1}^d \Cbf_i, \\ 
    \Cbf_i = \Exp\left[\overline{\zbf}_{:i} \overline{\ybf}_{:i}\herm \right] \Exp\left[\overline{\ybf}_{:i} \overline{\ybf}_{:i}\herm \right]^{-1} \Exp\left[\overline{\ybf}_{:i} \overline{\zbf}_{:i}\herm \right],
\end{gather}
\end{subequations}
where $\overline{\Zbf} = \Zbf-\Exp[\Zbf]$ and $\overline{\Ybf} = \Ybf-\Exp[\Ybf]$. Then, the mutual information between $\Zbf$ and $\Ybf$ is bounded below by,
\[
    I(\Zbf;\Ybf) \geq  -d \ln(|\Ibf- \Pbf_C|).
\]
\end{lemma}
\begin{proof}
 See Appendix \ref{sec:rate_lin_proof}. 
\end{proof}

Assume that the rows of the noise matrix $\xibf_{n:}$ are i.i.d.\ with some distribution
$ \Xi$ where $\Exp|\Xi|^2 < \infty$.
Then, using \eqref{eq:tau_alpha_gen} 
calculate $\Abf_{\rm rx}$ and $\Tbf_{\rm rx}$ with $\Phi(U) =\Wbf_{\rm rx}Q_{\rm rx}(F(\Wbf_{\rm tx}Q_{\rm tx}(U), \Xi))$, where $U \sim C\mathcal{N}(\zero, \Pbf),  ~ \Pbf = \sum_{m}\delta_m \Pbf_m$.

\begin{theorem} \label{thm:rate_lin}
Under the assumptions in Sec.~\ref{subsec:assump} and given \eqref{eq:zmix} and \eqref{eq:Pbar}, the linear rate is almost surely bounded below by,
\beq \label{eq:rate_lin}
    R_{\rm lin} \geq 
    \sum_{m=1}^M \delta_m \log\left(\left|\Ibf + \Abf_{\rm rx} \Pbf_{m}\Abf_{\rm rx} \herm \Tbf_{\rm rx}^{-1} \right| \right).
\eeq
\end{theorem}
\begin{proof}
See Appendix \ref{sec:rate_lin_proof}. 
\end{proof}

The lower bound can be interpreted as follows. For Gaussian inputs, the system in Fig. \ref{fig:Sys_Model} can be equivalently modeled with a MIMO system with channel matrix gain $\Abf_{\rm rx}$ where the channel outputs are added with a Gaussian noise vector of covariance $\Tbf_{\rm rx}$. The lower bound is basically the Shannon capacity of this equivalent linear system when the per sub-band covariance matrices are fixed. More precisely, for the case of MIMO channel with one sub-band and no quantization the lower bound leads to 
\beq 
    R_{\rm lin} \geq \log\left(\left|\Ibf + \frac{1}{\sigma^2}\Hbf \Pbf\Hbf \herm \right| \right)
\eeq
which is the Shannon capacity of MIMO channel $\Hbf$ with the transmit covariance matrix $\Pbf$. 

\section{Simulations and Numerical Results}
\label{sec:num_res}

\begin{figure*}[t]
\centering
\begin{subfigure}{0.31\linewidth}
    \centering
    \includegraphics[width=\textwidth]{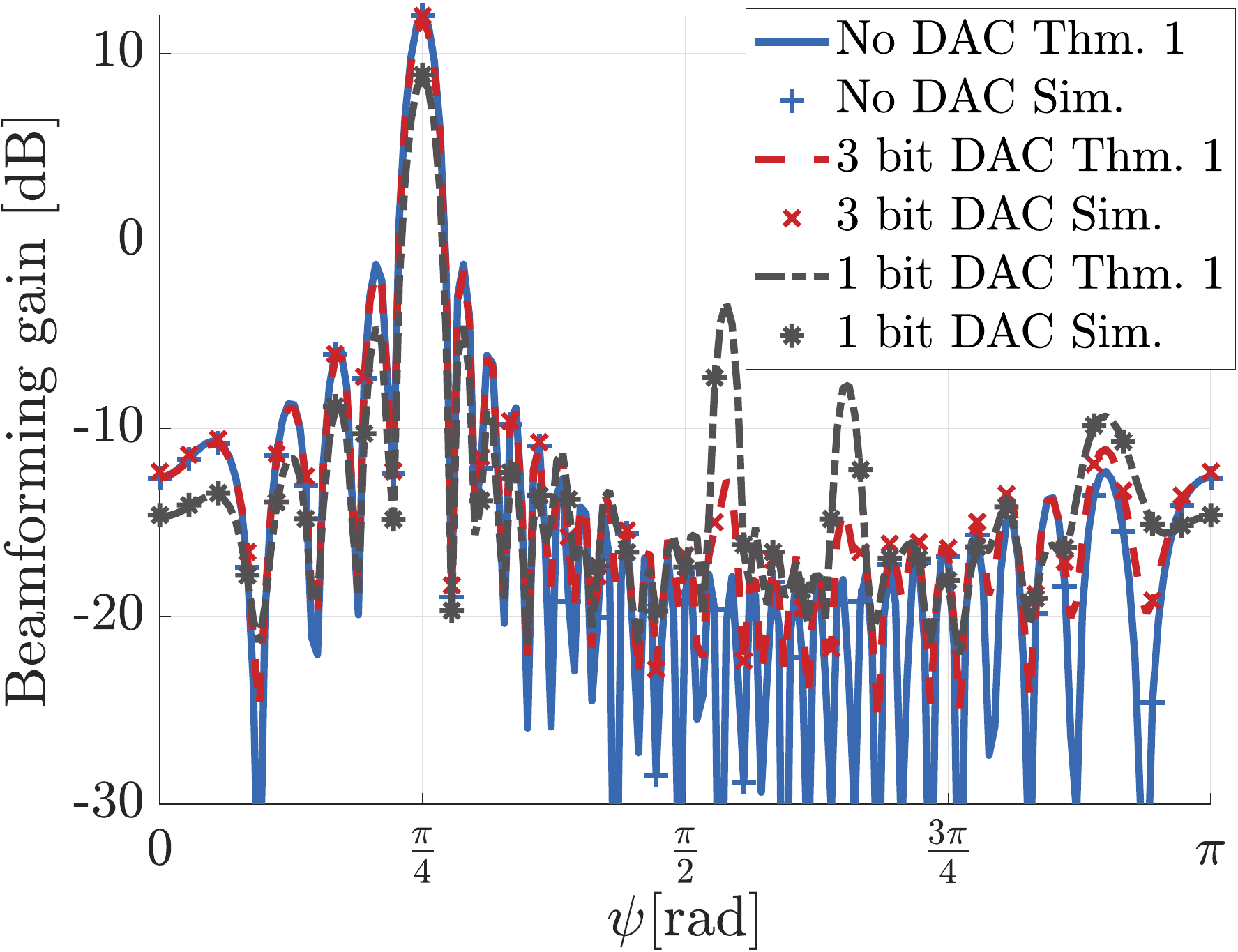}
    \caption{}
    \label{fig:EnergyVsRes}
\end{subfigure}
\begin{subfigure}{0.31\linewidth}
    \centering
    \includegraphics[width=\textwidth]{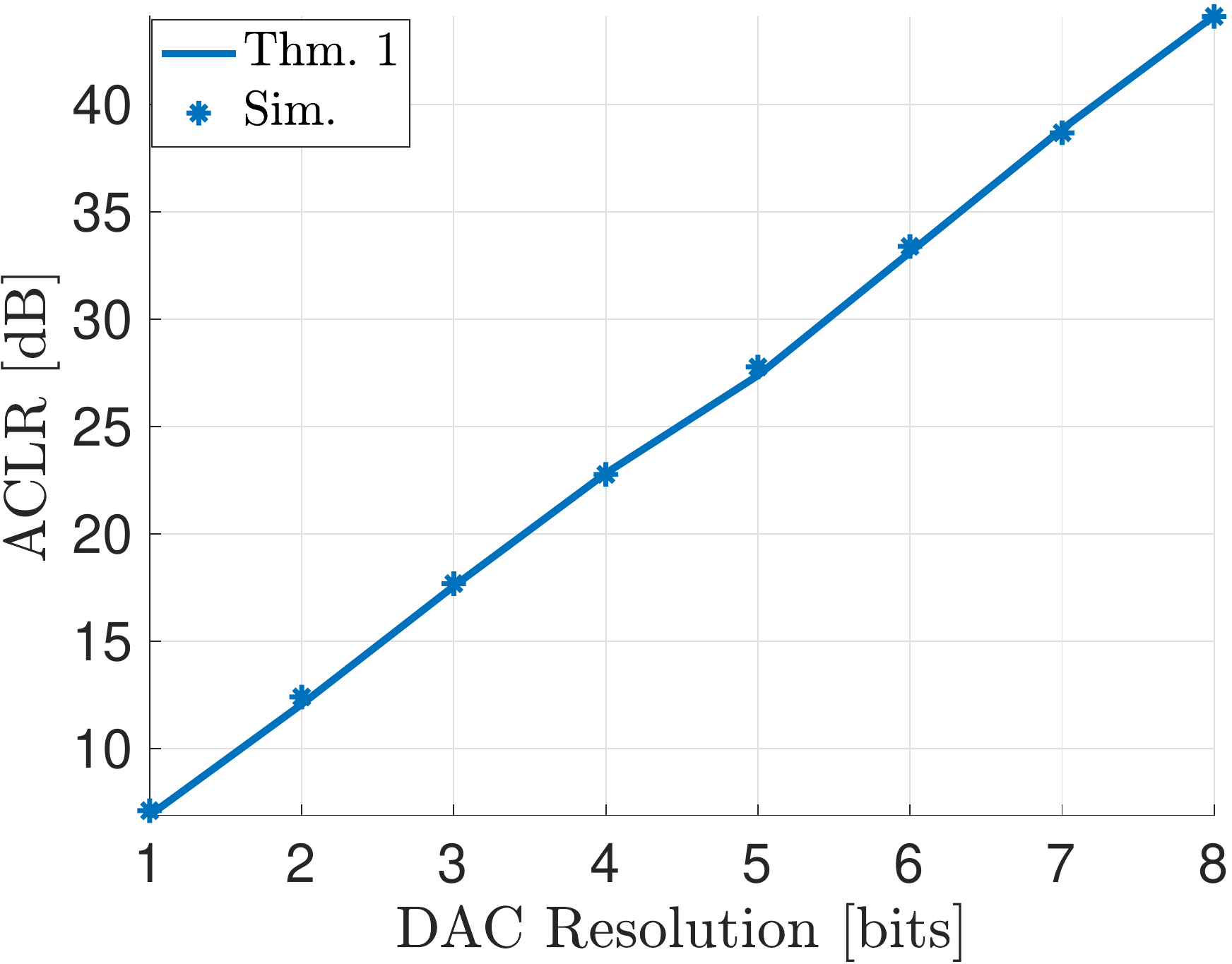}
    \caption{}
    \label{fig:ACLRVsRes}    
\end{subfigure}
\begin{subfigure}{0.31\linewidth}
    \centering
    \includegraphics[width=\textwidth]{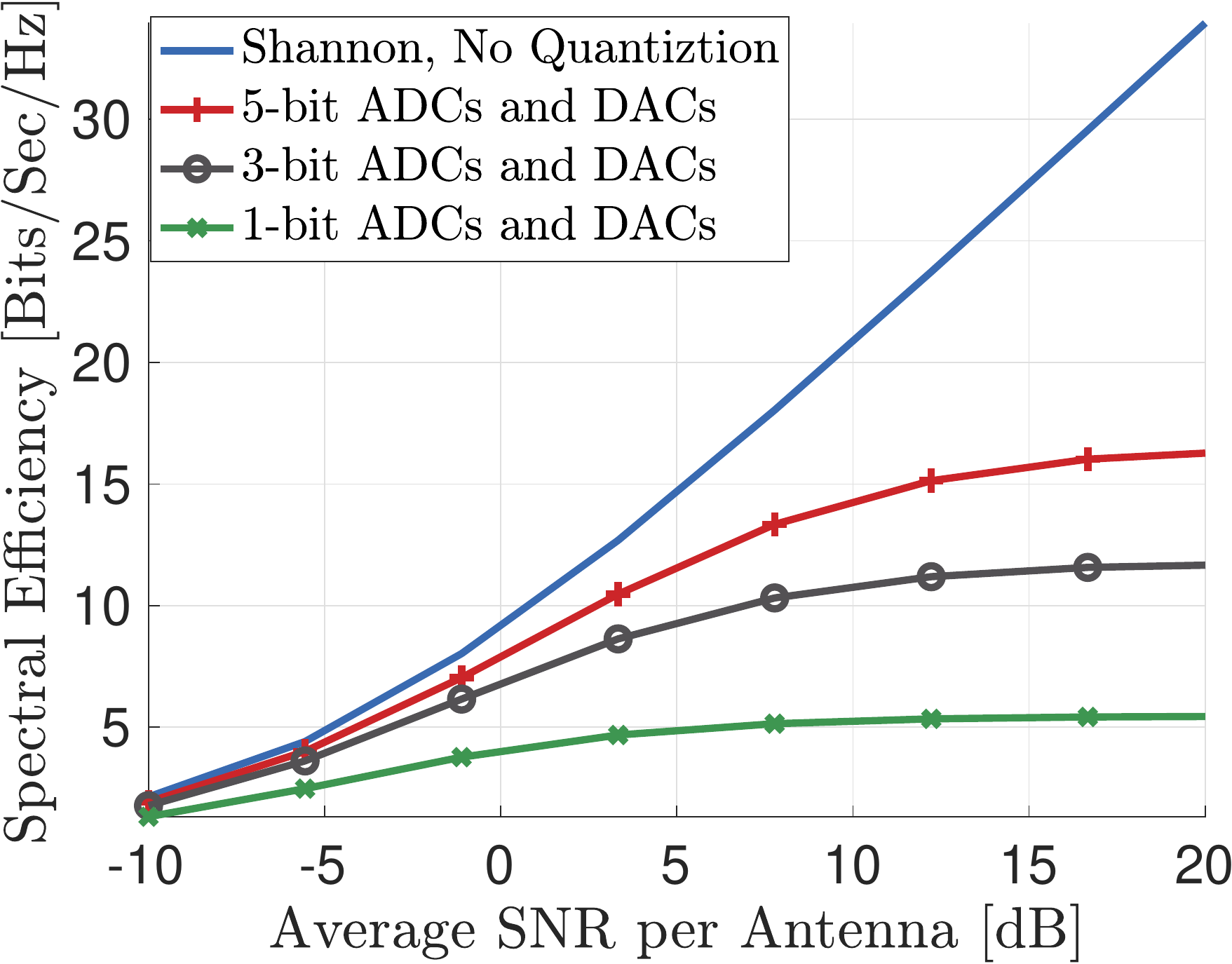}
    \caption{}
    \label{fig:RateVsSNR}
\end{subfigure}
\vspace*{-0.1cm}
\caption{(a) Spatial transmit power, ${\rm BF}_m(\cdot)$, for the transmit sub-band considering different DAC resolutions based on \eqref{eq:bfg} for $\phi = \pi/4$. (b)  ACLR based on  \eqref{eq:ACLR} for different DAC resolutions. (c) Achievable rate of the system based on Thm.~\ref{thm:rate_lin} for different values of average received SNR per antenna and DACs and ADCs resolutions.}
\vspace*{-0.4cm}
\end{figure*}

For our simulations, we consider a MIMO system, where both transmitter and receiver use ULAs with half wavelength spacing and $16$ transmit and eight receive antennas equipped with low resolution uniform DACs and ADCs, respectively. Moreover, for each quantizer, its spacing is optimized to minimize the mean square error distortion between its input and output according to \cite{max1960quantizing}. We assume that the analog precoder and combiner are identity matrices. Also, we use a Rayleigh fading channel whose elements are generated using i.i.d. symmetric complex Gaussian distribution with variance one and mean zero (i.e., $C{\cal N}(0,1)$). We consider that there are two sub-bands, $m=2$, and the transmitter is transmitting in the first sub-band only (i.e., $\Pbf_2 = \mathbf{0}$). For evaluating matrices $\Abf$ and $\Tbf$ for Thm.~\ref{thm:spec_fd} and Thm.~\ref{thm:rate_lin}, we use $1024$ point FFT (i.e., $N  = 1024$). More precisely, we assume that the first $512$ samples correspond to the desired transmit band and the rest correspond to the adjacent band. Furthermore, we average the simulations over $100$ channel realizations.

\paragraph*{Spectrospatial power}  
First, we validate the accuracy of the per sample power formula in \eqref{eq:sm_fd} from Thm.~\ref{thm:spec_fd}. Let us consider that the transmitter performs digital beamforming to direct its transmit stream towards angle $\phi$. More precisely, to generate $\zbf_{n:}, n\in[512]$, it uses the covariance matrix 
\begin{equation}
    \label{eq:simcov}
    \Pbf_1 = \frac{1}{N_t}\ebf(\phi)\ebf(\phi)\herm, 
\end{equation}
where $\ebf(\phi) = [1, e^{i\pi \cos(\phi)}, \ldots, e^{(N_t-1)i\pi \cos(\phi)}]\tran$ is the array response of the transmitter's ULA. 
In Fig.~\ref{fig:EnergyVsRes}, we have plotted the spatial transmit power, ${\rm BF}_m(\psi)$ for $\psi\in(0,\pi]$ using \eqref{eq:bfg} for $\phi = \pi/4$ considering different DAC resolutions when $\Sbf_1$ is calculated based on Thm.~\ref{thm:spec_fd} and when it is evaluated through simulation. We observe that the theorem accurately predicts the simulations. Furthermore, as the DAC resolution decreases the main lobe gain is reduced while side lobes' gains are increased, leaking more power into the side-lobes. 

To evaluate our model's accuracy for the spectral power distribution, let us define the ACLR of this transmitter as ratio of the transmit power in the first sub-band over the leaked power in the second sub-band. Based on \eqref{eq:nu_fd}, we have
\begin{align}
\label{eq:ACLR}
    {\rm ACLR} = \frac{\nu_1}{\nu_2}.
\end{align}

In Fig.~\ref{fig:ACLRVsRes}, we have plotted the ACLR for the covariance matrix in \eqref{eq:simcov} with $\phi = \pi/4$ considering different DAC resolutions when $\nu_1$ and $\nu_2$ are calculated using Thm.~\ref{thm:spec_fd} and when they are evaluated through simulation. As observed, the theorem accurately predicts the simulations and as expected the ACLR is an increasing function of the DAC resolution.

\paragraph*{Achievable rate of the system} 

Next, we use Thm.~\ref{thm:rate_lin} to plot the achievable rate of the system. For this case we use the identity matrix for $\Pbf_1$. The rate lower bound considering DAC and ADC resolutions of one, three, and five bits are plotted based on \eqref{eq:rate_lin} from Thm.~\ref{thm:rate_lin} in Fig.~\ref{fig:RateVsSNR}. As a performance benchmark and upper bound, we have also plotted the Shannon capacity of the system corresponding to infinite resolution DACs and ADCs. As expected, increasing the resolution of DACs and ADCs improves the rate. We also observe that the system performs very close to the upper bound at low SNRs.

\section{Conclusion}
We have studied quantized MIMO systems with linear modulation and provided a rigorous, yet simple equivalent linear model in the limit of large random transformations for modulation. The provided model accurately captures the effect of quantization on the spectrospatial power distribution and achievable rate of the system. We have also validated the spectrospatial predictions through simulations and numerical evaluations. 
\clearpage
\balance
\bibliographystyle{IEEEtran}
\bibliography{bibl}

\clearpage
\appendices  

\section{Proof of Proposition \ref{prop:matdist}}
\label{app:matdist}

For the proof of this proposition, we first prove the following results on left circularly symmetric matrices and functional equations.

\begin{lemma}
\label{lem:quad}
Let $\Vbf \in \mathbb{C}^{N\times N}$ be a left spherical symmetric matrix. Then, its characteristic function $\Phi(\Tbf)$, satisfies the followings
\begin{enumerate}
    \item $\Phi(\Tbf) = \Phi(\Vbf\herm \Tbf)$ for $\Vbf \in O(N)$, where $O(N)$ denotes the set of $N \times N$ orthogonal matrices.
    \item $\Phi(\Tbf)$ is a function of $\Tbf\herm \Tbf$.
\end{enumerate}
\end{lemma} 
The proof is a straightforward extension of the proof of \cite[Theorem 2.1]{fang2018symmetric}.  

\begin{lemma}
\label{lem:func}
Let $\Fbf(\abf) = \Fbf(\bbf)\Fbf(\cbf)$ ,where $\Fbf(\tbf)$ is a mapping from $\mathbb{C}^{d}$ to $\mathbb{C}$ and a function of $\tbf\tbf\herm$ and $\abf \abf \herm = \bbf \bbf \herm+ \cbf \cbf \herm$. Therefore, 
\beq 
\Fbf(\abf) = e^{-\frac{1}{4} \abf \herm \Pbf \abf},
\eeq
where $\Pbf \in \mathbb{C}^{d \times d}$.
\end{lemma}
\begin{proof}
Since $\Fbf(\abf)$ is a function of $\abf\abf\herm$, we can write it as $\Fbf(\tbf) = \Gbf(\tbf\herm \tbf)$. Set $\Abf = \abf\abf\herm, \Bbf = \bbf\bbf\herm$, and $\Cbf = \cbf\cbf\herm$ we have: 
\beq 
\Gbf(\Abf) = \Gbf(\Bbf)\Gbf(\Cbf) 
\eeq
where $\Gbf(\Tbf)$ is a mapping from $\mathbb{C}^{d\times d}$ to $\mathbb{C}$. Therefore, 
\beq 
\label{eq:func1}
\Gbf(\Tbf) = e^{\Jbf(\Tbf)},  
\eeq
where $\Jbf(\Tbf)$ is  mapping from $\mathbb{C}^{d\times d}$ to $\mathbb{C}$ and $\Jbf(\Abf) =  \Jbf(\Bbf)+ \Jbf(\Cbf)$. Therefore,
\begin{align}
\begin{aligned}
    \!\!\Jbf(\Tbf) &= \sum_{j} {\pbf_l}_j\tran \Tbf {\pbf_r}_j  \sum_{j}\Tr\left ( {\pbf_l}_j\tran \Tbf {\pbf_r}_j\right) =  \sum_{j}\Tr\left(\Tbf {\pbf_r}_j {\pbf_l}_j\tran\right)\\
    &= \Tr\left(\sum_{j}\Tbf {\pbf_r}_j {\pbf_l}_j\tran\right)\!\! =\! \Tr\left(\!\Tbf\left\{\sum_{j} {\pbf_r}_j {\pbf_l}_j\tran\right\}\!\right)\!. 
\end{aligned}
\end{align}
As a result, we can write 
\beq 
\label{eq:func2}
\Jbf(\Tbf) = -\frac{1}{4} \tr\left(\Pbf\Tbf\right),  
\eeq
Where $\Pbf \in \mathbb{C}^{d\times d}$. Using $\Tbf = \tbf\tbf \herm$ and substituting \eqref{eq:func2}  in \eqref{eq:func1}, we have 
\beq 
\Fbf(\tbf) = e^{-\frac{1}{4} \tr\left(\Pbf\tbf\tbf \herm\right)} = e^{-\frac{1}{4}\tr\left(\tbf \herm\Pbf\tbf\right)} = e^{-\frac{1}{4}\tbf \herm\Pbf\tbf}.  
\eeq
Which concludes the proof.
\end{proof}

To prove the proposition, we use a similar proof as of the paper \cite{kingman1972random}. 
Note that the matrix $\Ubf$ is left-spherically symmetric. Therefore, any rotation of the rows of matrix $\Ubf$ would preserve its distribution. Hence, using Finetti's theorem, there is a $\sigma-field, \Psi$  of events conditional upon which the $\ubf_{n:}, n\in[N]$ are independent and have
the same distribution function $F(\ubf)$. Define, 
\begin{align}
\begin{aligned}
    \Phi(\tbf) &= \int e^{i \Re\left ( \tr \left( \tbf\herm \ubf\right) \right)} d F(\ubf) = \Exp \left[e^{i \Re\left ( \tr \left( \tbf\herm \ubf\right) \right)}\right|\Psi \Big],
\end{aligned}
\end{align}
where $\tbf \in \mathbb{C}^{d}$.

The conditional independence means that, for matrix $\Tbf \in \mathbb{C}^{N\times d}$,
\begin{align}
    &\Exp \left[ e^  {i \Re\left ( \tr \left( \Tbf\herm \Ubf\right) \right)}  \Big|\Psi \right] = \prod_{n=1}^N \Phi(\tbf_r) \label{eq:ch1}\\
    &\Exp \left[ e^ { i \Re\left ( \tr \left( \Tbf\herm \Ubf\right) \right) }\right] = \Exp \left[\prod_{n=1}^N \Phi(\tbf_r) \right] \label{eq:ch2}
\end{align}
From Lemma \ref{lem:quad}, we know that the left and therefore the right side of \eqref{eq:ch2} are function of $\Tbf \herm \Tbf = \sum_{r= 1}^N \tbf_r\herm \tbf_r$. 
Let $\tbf\herm \tbf = \hbf \herm \hbf + \gbf \herm \gbf$. Then,
\begin{align}
\label{eq:chareq}
\begin{aligned}
    &\Exp \left[|\Phi(\tbf) - \Phi(\hbf)\Phi(\gbf)|^2 \right] = \Exp[\Phi(\tbf)\Phi(-\tbf)]\!\\
    &+\! \Exp[\Phi(\hbf)\Phi(\gbf)\Phi(-\hbf)\Phi(-\gbf)] 
    - \Exp[\Phi(\tbf)\Phi(\hbf)\Phi(\gbf)]\!\\
    &- \!\!\Exp[\Phi(-\tbf)\Phi(-\hbf)\Phi(-\gbf)].
\end{aligned}
\end{align}
All four elements in the right hand side of \eqref{eq:chareq} are equal since each element is in the form of \eqref{eq:ch2} and is function of $2\tbf \herm \tbf$. Therefore,
\beq
    \Phi(\tbf) \stackrel{a.s}{=} \Phi(\hbf)\Phi(\gbf)
\eeq 
Based on Lemma \ref{lem:func}, the solution to this functional equation is 
\beq
    \Phi(\tbf \herm) = e^{-\frac{1}{4} \tbf \herm \Pbf \tbf}.
\eeq 
Using \eqref{eq:ch1},
\begin{align}
\begin{aligned}
    \Exp& \left[ e^ { i \Re\left ( \tr \left( \Tbf\herm \Ubf\right) \right) } \Big|\Pbf \right] \\
    &= \Exp \left[\prod_{n=1}^N e^{-\frac{1}{4} \tbf_n \herm \Pbf \tbf_n} \Big|\Pbf \right] = \prod_{n=1}^N  e^{-\frac{1}{4} \tbf_n \herm \Pbf \tbf_n}
\end{aligned}
\end{align}

Therefore, conditioned on $\Pbf$, $\ubf_{n:}$ are i.i.d.  $\sim \mathcal{CN}(\zero, \Pbf)$. We can find $\Pbf$ as follows

\begin{align}
    \begin{aligned}
    \Pbf &\stackrel{(a)}{\stackrel{a.s}{=}} \frac{1}{N} \sum_{n=1}^N \ubf_{n:}  \ubf_{n:}\herm  = \frac{1}{N}  \Ubf \tran \Ubf\str \\
    &\stackrel{(b)}{=} \frac{1}{N} \Xbf \tran  \Xbf\str = \frac{1}{N} \sum_{n=1}^N \xbf_{n:}\xbf_{n:}\herm  \stackrel{(c)}{\stackrel{a.s}{=}} \Exp[XX\herm],
    \end{aligned}
\end{align}
where $(a)$ follows from strong law of large number, $(b)$ is due to equality $\Ubf = \Vbf \Xbf$, and $(c)$ comes from empirical convergence of rows of $\Xbf$ to random vector $X\tran$.

\section{Proof of Proposition \ref{prop:lin}}
\label{pr:emp_conv}
For the proof of this proposition, we use an approach similar to the proof of \cite[Theorem 4]{rangan2019vector}. Consider the following steps:
\begin{subequations}
\begin{align}
    \Ubf &= \Vbf \herm \Zbf,\\
    \Abf &= \bkt{\frac{\partial \Phibf(\Ubf,\Xibf)}{\partial U}}, \\
    \Bbf &= \Phibf(\Ubf,\Xibf)- \Ubf \Abf\tran , \\
    \Theta&= \Vbf \Bbf,
\end{align}
\end{subequations}
where $\bkt{\frac{\partial \Phibf(\Ubf,\Xibf)}{\partial U}}$ is the Jacobian matrix of $\Phibf(\cdot)$ evaluated for each row pair of $(\ubf_{n:},\xibf_{n:})$ and then averaged over $n$.
The transpose of rows of matrix $\Xibf$ (i.e., $\xibf_{n:}$) are i.i.d. $\mathcal{CN}(0,\Ibf_{N_r})$. Therefore,
\begin{align}
        \lim_{n \rightarrow \infty} \{\xibf_{n:}\} \stackrel{PL(2)}{=} \xi,\quad \xi \sim \mathcal{CN}(\zero, I_{N_r})
\end{align}
 From Prop.~\ref{prop:matdist}, we know that the rows of matrix $\Ubf$ converge empirically to random vector $U \sim \mathcal{CN}(\zero, \Pbf)$. Additionally, since $\Vbf$ and $\Zbf$ are independent of $\Xibf$ and the function $\Phibf(\cdot)$ is Lipschitz continuous, the sequence $(\xibf_{n:}, \Ubf_{n:}, \bbf_{n:})$ converges empirically to $(\Xi, U, B)$. Following same steps of \cite[Theorem 4]{rangan2019vector}, we have 
\begin{align}
\Theta &= \Theta^{det} + \Theta^{rand},
\end{align}
where $\Theta^{det} = \zero$ and the rows of the matrix $\Theta^{rand}$, converge empirically to $\Theta \sim \mathcal{CN}(\zero, \Tbf)$ where $\Tbf = \Exp\left[(\Phi(U,\Xi)- \Abf U ) (\Phi(U,\Xi)- \Abf U)\herm \right]$.

\section{Proof of Theorem \ref{thm:spec_fd}} \label{sec:spec_proof}
The theorem is a direct application of the linear model in Prop.~\ref{prop:lin}.
To use the proposition, first observe that, due to \eqref{eq:delm_limit} and the Gaussian distribution on $\zbf_{n:}$ in \eqref{eq:zmix},
we have that the elements of sub-band selection $\lbf$ and the rows of frequency-domain inputs $\Zbf$ converge empirically as,
\beq \label{eq:alim}
    \lim_{n \rightarrow \infty} \{(\zbf_{n:},\lbf_n)\} \stackrel{PL(2)}{=} (Z,\ell),
\eeq
where $\ell \in [M]$ 
is a discrete random variable with $\Pr(\ell=m)=\delta_m$ and $Z$ is 
the conditional complex Gaussian, 
\[
    Z \sim C{\mathcal N}(\zero,\Pbf_m) \mbox{ when } \ell=m.
\]
In particular, the covariance matrix of $Z$ is,
\beq \label{eq:expz_sm}
    \Exp[ZZ\herm] = \sum_{m=1}^M \delta_m \Pbf_m =: \Pbf.
\eeq
Now, the frequency domain components of the transmitted vectors $\xbf_{:j}, j \in [N_t]$ are given by,
\[
    \Rbf = \Vbf\Xbf = \Vbf\Qbf(\Vbf\herm \Zbf).
\]
Proposition~\ref{prop:lin} then shows that the components of $(\rbf,\zbf,\lbf)$ converge empirically as,
\[
    \lim_{N \rightarrow \infty} \{(\rbf_{n:},\zbf_{n:},\ell_n) \} \stackrel{PL(2)}{=} (R,Z,\ell), 
\]
and
\[
    R = \Abf_{\rm tx}Z + W_{\rm tx}, \quad W_{\rm tx} \sim C{\mathcal N}(\zero,\Tbf_{\rm tx}),
\]
where $W_{\rm tx}$ is independent of $Z$.  The sub-band energies,
\begin{align}
    \Sbf_m &:= \lim_{N \rightarrow \infty}\frac{1}{N} \sum_{n=1}^{N} \rbf_{n:} \rbf_{n:}\herm \indic{\ell_n = m} = \Exp\left[ R R\herm \indic{\ell = m} \right] \nonumber \\
    &= \Exp\left[ (\Abf_{\rm tx}Z + W_{\rm tx})(\Abf_{\rm tx}Z + W_{\rm tx})\herm |\ell=m\right]\Pr(\ell=m) \nonumber \\
    &= \left[ \Abf_{\rm tx}\Pbf_m\Abf_{\rm tx} \herm + \Tbf_{\rm tx} \right] \delta_m \nonumber.
\end{align}
This proves \eqref{eq:sm_fd}.
To prove \eqref{eq:stot_fd},
\begin{align}
\begin{aligned}
    \Sbf &= \sum_{m=1}^M \Sbf_m = \sum_{m=1}^M \left[  \Abf_{\rm tx}\Pbf_m\Abf_{\rm tx} \herm + \Tbf_{\rm tx} \right] \delta_m \\
    &= \Abf_{\rm tx}\Pbf\Abf_{\rm tx} \herm + \Tbf_{\rm tx},
    \end{aligned}
\end{align}
where the last step used \eqref{eq:expz_sm} and the fact that $\sum_m \delta_m = 1$.

\section{The Linear Rate Region} \label{sec:linear_feas}
$(a) \Rightarrow (b)$:
Suppose there exists $\Pbf_m $ as in (a) and let $\nu_m$
be given by \eqref{eq:nu_fd}.
Since $\Pbf = \sum_m \delta_m \Pbf_m$, we have $\sum_m \nu_m =1$.  Also, since $\diag(\Pbf_m) \geq 0$, we have $\nu_m$ in \eqref{eq:nu_fd} satisfies the lower bound \eqref{eq:nu_fd_feas}.

$(b) \Rightarrow (a)$: Conversely, suppose we are given $\nubf$ satisfying
 \eqref{eq:nu_fd_feas} with $\sum_m \nu_m = 1$.
Set,
\begin{align}
 \label{eq:Pm_nu}
 \begin{aligned}
    \Pbf_m =\frac{\Pbf}{\tr \left( \Abf_{\rm tx}\Pbf\Abf_{\rm tx} \herm \right)}\left[\frac{\nu_m}{\delta_m}\tr \left( \Abf_{\rm tx}\Pbf\Abf_{\rm tx} \herm + \Tbf_{\rm tx} \right) - \tr(\Tbf_{\rm tx})\right].
 \end{aligned}
\end{align}
Therefore, $\nu_m$ satisfies \eqref{eq:nu_fd}.
Since $\nu_m$ satisfies \eqref{eq:nu_fd_feas}, 
$P_m$ in \eqref{eq:Pm_nu} satisfies $P_m \geq 0$. Also,
\begin{align*}
    \sum_m &\delta_m P_m = \frac{\Pbf}{\tr \left( \Abf_{\rm tx}\Pbf\Abf_{\rm tx} \herm \right)}\left[ \sum_m \nu_m \tr \left( \Abf_{\rm tx}\Pbf\Abf_{\rm tx} \herm + \Tbf_{\rm tx} \right)\right.\\
    &- \left.\sum_m \delta_m \tr(\Tbf_{\rm tx}) \right] \\
    &= \frac{\Pbf}{\tr \left( \Abf_{\rm tx}\Pbf\Abf_{\rm tx} \herm \right)} \left[  \tr \left( \Abf_{\rm tx}\Pbf\Abf_{\rm tx} \herm+ \Tbf_{\rm tx}- \Tbf_{\rm tx}  \right) \right] = \Pbf,
\end{align*}
where we have used the fact that $\sum_m \nu_m = \sum_m \delta_m = 1$.

\section{Proof of Theorem \ref{thm:rate_lin}} \label{sec:rate_lin_proof}

We need two basic mutual information lemmas.
For $m \in [M]$, let $\Zbf[m]$ and $\Zbfhat[m]$ denote the sub-matrices of $\Zbf$ and $\Zbfhat$
with components in sub-band $m$.  

\begin{lemma} \label{lem:chain}  The mutual information is bounded below by,
\beq \label{eq:mi_sum}
    I(\Zbf;\Zbfhat) \geq \sum_{m=1}^M I(\Zbf[m];\Zbfhat[m]).
\eeq
\end{lemma}
The proof follows using the same steps as of the proof of \cite[Lemma 1]{dutta2020capacity}.
\begin{lemma} \label{lem:mi_lower}  Suppose that $\Zbf \in \C^{N_t\times d}$ is a random matrix 
with i.i.d. columns $Z \sim  C{\mathcal N}(\mubf_{\zbf},
\Pbf)$.  Let $\Ybf$ be another random matrix and define,
\begin{subequations}
\begin{gather}
    \Pbf_C := \Pbf^{-1}\frac{1}{d}\sum_{ i=1}^d \Cbf_i, \\ 
    \Cbf_i = \Exp\left[\overline{\zbf}_{:i} \overline{\ybf}_{:i}\herm \right] \Exp\left[\overline{\ybf}_{:i} \overline{\ybf}_{:i}\herm \right]^{-1} \Exp\left[\overline{\ybf}_{:i} \overline{\zbf}_{:i}\herm \right],
\end{gather}
\end{subequations}
where $\overline{\Zbf} = \Zbf-\Exp[\Zbf]$ and $\overline{\Ybf} = \Ybf-\Exp[\Ybf]$. Then, the mutual information between $\Zbf$ and $\Ybf$ is bounded below by,
\[
    I(\Zbf;\Ybf) \geq  -d \ln(|\Ibf- \Pbf_C|).
\]
\end{lemma}
\begin{proof}
The mutual information is,
\beq \label{eq:Igauss}
    I(\Zbf;\Ybf) = H(\Zbf) - H(\Zbf|\Ybf).
\eeq
Since $\tilde{\zbf} ={\rm vec}(\Zbf) = [\zbf_{:1}\tran,\zbf_{:2}\tran,...,\zbf_{:d}\tran]\tran$ is distributed as $C{\mathcal N}({\rm vec}(\mubf_{\zbf}),\Ibf_d \otimes \Pbf)$, 
\beq \label{eq:Hgauss}
    H(\Zbf) = d\ln( |2\pi e \Pbf|).
\eeq
Define $\tilde{\ybf} = {\rm vec}(\Ybf)=[\ybf_{:1}\tran,\ybf_{:2}\tran,...,\ybf_{:d}\tran]\tran$. Now, given $\Ybf$, we can get a estimate of $\tilde{\zbf}$, $\hat{\tilde{\zbf}}$ which leads to the error covariance matrix,
\beq
    \Sigma_e := \Exp\left[ (\tilde{\zbf} - \hat{\tilde{\zbf}}(\tilde{\ybf}))(\tilde{\zbf} - \hat{\tilde{\zbf}}(\tilde{\ybf}))\herm\right],
\eeq
where $\hat{\tilde{\zbf}}(\tilde{\ybf})$ is the estimate of $\tilde{\zbf}$ given $\tilde{\ybf}$.
So, the conditional entropy $H(\Zbf|\Ybf)$ is bounded below by the entropy of a Gaussian random vector with the covariance $\Sigma_e$. Consider linear minimum mean square estimation and we have
\begin{align}
    &H(\Zbf|\Ybf) \leq \ln((2\pi e)^{d N_t}|\Sigma_e|)\\
    &\stackrel{(a)}{\leq}  \ln((2\pi e)^{d N_t} \prod_{i = 1}^d|\Cbf_{ei}|) \\
    &=  d \ln((2\pi e)^{N_t})  + N_t\ln( \prod_{i = 1}^d|\Cbf_{ei}|^{\frac{1}{N_t}}) \\
    &\stackrel{(b)}{\leq}  d \ln((2\pi e)^{N_t})  + d N_t\ln(\frac{1}{d}\sum_{i = 1}^d\left|\Cbf_{ei}\right|^{\frac{1}{N_t}})\\
    &\stackrel{(c)}{\leq} d \ln((2\pi e)^{N_t})  + d N_t\ln(\frac{1}{d}\left|\sum_{i = 1}^d\Cbf_{ei}\right|^{\frac{1}{N_t}})\\
    &= d \ln((2\pi e)^{N_t}\left|\frac{1}{d}\sum_{i = 1}^d\Cbf_{ei}\right|)
\end{align}
where $\Cbf_{ei} =$ $\Exp\left[\bar{\zbf}_{:i} \overline{\zbf}_{:i}\herm\right]$ $- \Exp\left[\overline{\zbf}_{:i} \overline{\ybf}_{:i}\herm \right]$ $\Exp\left[\overline{\ybf}_{:i} \overline{\ybf}_{:i}\herm \right]^{-1} \Exp\left[\overline{\ybf}_{:i} \overline{\zbf}_{:i}\herm \right]$, (a) uses Fischer's inequality \cite{fischer1908hadamardschen}, (b) uses Jensen's inequality, and (c) uses an extended version of the Hadamard's inequality \cite[Prop.~2.7]{li2020extensions}. 
Therefore, we have
\beq \label{eq:Hzy1}
    H(\Zbf|\Ybf) \leq d\ln(\left|2\pi e \left(\Pbf-\frac{1}{d}\sum_{ i=1}^d \Cbf_i\right)\right|).
\eeq
Substituting \eqref{eq:Hgauss} and \eqref{eq:Hzy1} into \eqref{eq:Igauss}, we get
\begin{align}
    &I(\Zbf;\Ybf) \geq \nonumber\\
    &d\ln(|2\pi e \Pbf|) - d \ln(\left|2\pi e \left(\Pbf-\frac{1}{d}\sum_{ i=1}^d \Cbf_i\right)\right|)\\
    &= -d\ln(|\Ibf - \Pbf_C|),
\end{align}
which concludes the proof.
\end{proof}

To prove the theorem, we use these lemmas as follows.  In each sub-band $m$, the rows of $\Zbf[m]$
are i.i.d.\ complex Gaussians with zero mean and covariance matrix $\Pbf_m$.
So, by Lemma~\ref{lem:mi_lower}, 
\beq \label{eq:Im1}
    I(\Zbf[m]; \Zbfhat[m]) \geq -N_m \ln(|\Ibf - \Pbf_{Cm}|),
\eeq
where $N_m$ is the number of coefficients in sub-band $m$ and
$\Pbf_{Cm}$ is the correlation coefficient matrix,
\begin{align}
\label{eq:rhom}
        \Pbf_{Cm} := 
            &\Pbf_m^{-1} \frac{1}{N_m} \sum_{i=1}^{N_m}  \Exp\left[\zbf[m]_{:i} \zbfhat[m]_{:i}\herm \right]\times\notag \\
        &\Exp\left[\zbfhat[m]_{:i} \zbfhat[m]_{:i}\herm \right]^{-1} \Exp\left[\zbfhat[m]_{:i} \zbf[m]_{:i}\herm \right].
\end{align} 
Now, \eqref{eq:delm_limit} shows that $N_m/N \rightarrow \delta_m$.
So, if we divide \eqref{eq:Im1} by $N$ and take the limit we get,
\beq \label{eq:Im2}
    \liminf_{N \rightarrow \infty} \frac{1}{N} I(\Zbf[m]; \Zbfhat[m]) \geq -\delta_m  \ln(|\Ibf - \overline{\Pbf}_{Cm}|).
\eeq
where $\overline{\Pbf}_{Cm}$ is the limiting correlation,
\beq \label{eq:rho_lim}
    \overline{\Pbf}_{Cm} := \lim_{N \rightarrow \infty} {\Pbf}_{Cm}
\eeq

To compute the limit in \eqref{eq:rho_lim}, we use 
a similar calculation to the proof of Theorem \ref{thm:spec_fd}.
Specifically, the received symbols are given by,
\beq
    \Zbfhat = \Vbf\Gbf(\Vbf\herm \Zbf,\Wbf).
\eeq
Proposition~\ref{prop:lin} then shows that the rows of $(\Zbfhat,\Zbf,\lbf)$ converge empirically as,
\beq
    \lim_{N \rightarrow \infty} \{(\widehat{Z}_{n:},Z_{n:},\ell_n) \} \stackrel{PL(2)}{=} (\widehat{Z},Z,L), 
\eeq
and
\beq
    \hat{Z} = \Abf_{\rm rx} Z + W_{\rm rx}, \quad W_{\rm rx} \sim {\mathcal CN}(\zero,\Tbf_{\rm rx}),
\eeq
where $W_{\rm rx}$ is independent of $Z$.  Now, we have that,
\begin{align}
    \label{eq:zzhat_lim1}
    \MoveEqLeft \lim_{N \rightarrow \infty} \Pbf_{Cm}= \begin{aligned}
        &\lim_{N_m \rightarrow \infty} \Pbf_m^{-1} \frac{1}{N_m} \sum_{i=1}^{N_m}  \Exp\left[\zbf[m]_{:i} \zbfhat[m]_{:i}\herm \right]\times\\
        &\Exp\left[\zbfhat[m]_{:i} \zbfhat[m]_{:i}\herm\right]^{-1} \Exp\left[\zbfhat[m]_{:i} \zbf[m]_{:i}\herm \right]
    \end{aligned}\\
    & = \begin{aligned}
        &\Pbf_m^{-1}\Exp\left[Z \hat{Z}\herm |L=m \right] \Exp\left[\hat{Z} \hat{Z} \herm|L=m  \right]^{-1}\times\\ 
        &\Exp\left[\hat{Z} Z |L=m \right] 
    \end{aligned}\\
    &= \left[\Abf_{\rm rx} \herm (\Abf_{\rm rx} \Pbf_m \Abf_{\rm rx} \herm  + \Tbf_{\rm rx})^{-1} \Abf_{\rm rx}\right] \Pbf_{m} \herm
\end{align}
Hence,
\begin{align}
    &(\Ibf - \overline{\Pbf}_C^{(m)})^{-1}\nonumber \\
    &\stackrel{(a)}{=}\Ibf + \Abf_{\rm rx} \herm \left( \Abf_{\rm rx} (\Pbf_m - \Pbf_m \herm) \Abf_{\rm rx} \herm + \Tbf_{\rm rx}\right)^{-1}\!\!\! \Abf_{\rm rx} \Pbf_{m} \herm \\
    &\stackrel{(b)}{=}  \Ibf + \Abf_{\rm rx} \herm \Tbf_{\rm rx}^{-1} \Abf_{\rm rx} \Pbf_{m}, 
\end{align}
where $(a)$ follows from Woodbury inversion lemma and $(b)$ uses the fact that $\Pbf_m$ is a covariance matrix. As a result, from \eqref{eq:Im2}, we obtain
\beq \label{eq:Im3}
\begin{aligned}
    \liminf_{N \rightarrow \infty} &\frac{1}{N} I(\Zbf[m]; \Zbfhat[m]) \geq \\
    &\delta_m  \log\left(| \Ibf + \Abf_{\rm rx} \herm \Tbf_{\rm rx}^{-1} \Abf_{\rm rx} \Pbf_{m} | \right).
\end{aligned}
\eeq
Substituting \eqref{eq:Im3} into the sum \eqref{eq:mi_sum} and using
Sylvester's determinant identity proves \eqref{eq:rate_lin}.

\end{document}